\documentclass[11pt]{article}
\usepackage{amssymb,amsmath,url,amsmath,graphicx,color}

\newcommand{\realrange}[2]{\left[#1, #2\right]}

\newcommand{\unitrange}[2]{\realrange{0}{1}}

\newcommand{\llabel}[1]{\label{\labelprefix:#1}}
\newcommand{\labelprefix}{} 

\newcommand{\discussionsize}{\small}

\newenvironment{code}{\noindent
\begin{tabbing}%
\hspace{1em}\=\hspace{1em}\=\hspace{1em}\=\hspace{1em}\=\hspace{1em}\=%
\hspace{1em}\=\hspace{1em}\=\hspace{1em}\=\hspace{1em}\=\hspace{1em}\=%
\kill}{\end{tabbing}}

\newcommand{\labelcommand}{}
\newcommand{\captiontext}{}
\newsavebox{\codeparam}
\newcounter{lineNumber}
\newenvironment{disscodepos}[3]{%
\renewcommand{\labelcommand}{#2}%
\renewcommand{\captiontext}{#3}%
\sbox{\codeparam}{\parbox{\textwidth}{#3}}%
\begin{figure}[#1]\begin{center}\begin{code}\setcounter{lineNumber}{1}}{%
\end{code}\end{center}\caption{\llabel{\labelcommand}\captiontext}\end{figure}}

{\end{disscodepos}}

\newcommand{\Do}       {{\bf do\ }}

\newcommand{\For}      {{\bf for\ }}
\newcommand{\Foreach}      {{\bf foreach\ }}

\newcommand{\If}       {{\bf if\ }}

\newcommand{\Then}     {{\bf then\ }}

\newcommand{\Return}   {{\bf return\ }}

\newdimen\endofsize\endofsize=0.5em
\def\endofbeweis{~\quad\hglue\hsize minus\hsize
                 \hbox{\vrule height \endofsize width
\endofsize}\par}

\def\RR{\mathbb R}

\def\EE{\mathbb E}

\def\cI{\mathcal I}
\def\cJ{\mathcal J}

\def\cL{\mathcal L}
\def\cN{\mathcal N}
\def\cS{\mathcal S}

\def\bT{\mathbf T}
\def\br{\mathbf r}

\def\OPT{\textsc{Opt}}

\def\TB{\textsc{Tb}}

\newcommand{\raf}[1]{(\ref{#1})}

\newcommand{\poly}{\operatorname{poly}}
\newcommand{\polylog}{\operatorname{polylog}}
\newcommand{\level}{\operatorname{level}}

\newcommand{\argmax}{\operatorname{argmax}}

\newcommand{\hide}[1]{}

\newcommand{\qed}{\hfill$\square$\bigskip}
\newcommand{\proof}{\noindent {\bf Proof}.~~}

\newtheorem{theorem}{Theorem}

\newtheorem{lemma}{Lemma}
\newtheorem{corollary}{Corollary}
\newtheorem{proposition}{Proposition}

\newtheorem{definition}[theorem]{Definition}

\def\xy{{$(x_i \vee x_j)$}}
\def\xcy{{$(\overline{x_i} \vee x_j)$}}
\def\xyc{{$(x_i \vee \overline{x_j})$}}
\def\xcyc{{$(\overline{x_i} \vee \overline{x_j})$}}

\date{}
\title{On Profit-Maximizing Pricing for the Highway and Tollbooth Problems}
\author{Khaled Elbassioni\thanks{Max-Planck-Institut f\"ur Informatik, Saarbr\"ucken, Germany; (\{elbassio,rraman\}@mpi-inf.mpg.de)}
\and Rajiv Raman\footnotemark[1] 
\and Saurabh Ray \thanks{Universit{\"a}t des Saarlandes, Saarbr{\"u}cken, Germany; (saurabh@cs.uni-sb.de)} 
\and Ren{\'e} Sitters\thanks{Department of Mathematics and Computer Science, VU, Amsterdam , the Netherlands; (rsitters@feweb.vu.nl)}
}

\begin{document}
\maketitle
\begin{abstract}
In the \emph{tollbooth problem}, we are given a tree $\bT=(V,E)$ with $n$ edges, and a set of $m$ customers, each of whom is interested 
in purchasing a path on the tree. Each customer has a fixed budget, and the objective is to price the edges of $\bT$ such that the total 
revenue made by selling the paths to the customers that can afford them is maximized. 
An important special case of this problem, known as the \emph{highway problem}, is when $\bT$ is restricted to be a line.

For the tollbooth problem, we present a randomized $O(\log n)$-approximation,
improving on the current best $O(\log m)$-approximation. 
We also study a special case of the tollbooth problem, when all the paths
that customers are interested in purchasing go towards a fixed root of $\bT$. In this case, we present an algorithm that returns
a $(1-\epsilon)$-approximation, for any $\epsilon > 0$, and runs in quasi-polynomial time. On the other hand, we rule out the existence of an FPTAS by showing that even for the line case, the problem is strongly NP-hard. Finally, we show that in the \emph{coupon model}, when we allow some items to be priced below zero to improve the overall profit, the problem becomes even APX-hard.
\end{abstract}
\section{Introduction}
Consider the problem of pricing the bandwidth along the links of a network such that the revenue obtained from customers interested 
in buying bandwidth along certain paths in the network is maximized. Suppose that each customer declares a set of paths 
she is interested in buying, and a maximum amount she is is willing to pay for each path. 
The network service provider's objective is to assign single prices to the 
links such that the total revenue from customers who can afford to purchase their paths is maximized. 
Recently, numerous papers have appeared on the computational complexity of such pricing
problems~\cite{AH06,BB07,BK06,BK07,CS08,DHFS06,ESZ07,GLSU06,GRR06,GHKKKM05,HK05,CS08}. 

A special case of this problem, where each customer is interested in purchasing only a single path (\emph{single-minded}), and where
there is no upper bound on the number of customers purchasing each link (\emph{unlimited supply}) was studied
by Guruswami et al. \cite{GHKKKM05}, under the name of \emph{tollbooth problem}. The authors of \cite{GHKKKM05} showed that the 
problem is already APX-hard when the network is restricted to be a tree, and
also presented a polynomial time algorithm for the case when all paths start at a certain root of the tree. In \cite{GHKKKM05}, the
authors also studied the \emph{highway problem}, a further restriction where the tree is a path, and gave polynomial time algorithms when either the budgets are bounded and integral, or all paths have a bounded length. 

In this paper, we continue the study of these problems. For the tollbooth problem, 
the best known approximation factor was $O(\log n+\log m)$, where $n$ and $m$ are respectively the number of 
edges of the tree and the number of customers. This result 
applies in fact for general sets \cite{GHKKKM05}, and not necessarily paths of a network, and even in the non single-minded case \cite{BBM08}. 
Very recently, and more generally, Cheung and Swamy \cite{CS08} gave an algorithm that, given any LP-based $\alpha$-approximation algorithm for 
maximizing the social welfare under limited supply, returns a solution with profit within a factor of $\alpha\log u_{max}$ of the maximum, 
where $u_{max}$ is the maximum supply of an item. In particular, this gives an $O(\log m)$-approximation for the tollbooth problem on trees. 
In this paper, we give an $O(\log n)$-approximation which is an improvement over the $O(\log m)$ since $n\leq 3m$ can be always assumed. 
We also show that if all the paths are going towards a certain root, then a $(1-\epsilon)$-approximation can be obtained in quasi-polynomial time. 
This result extends a recently developed quasi-PTAS \cite{ESZ07} for the highway problem, and uses essentially the same technique. However, there is a number of technical issues that have to be resolved for this technique to work on trees; most notably is the use of the Separator Theorem for trees, and the modification of the price-guessing strategy to allow only for \emph{one-sided} guesses.  

The existence of a quasi-PTAS for the highway problem indicates that a PTAS or even an FPTAS is still a possibility, since the problem was only known to be weakly NP-hard  \cite{BK06}. In the last section of this paper, we show that the highway problem is indeed strongly NP-hard and hence admits no FPTAS unless P=NP.

Balcan et al. \cite{BBCH07} considered a model in which some items can be priced below zero (in the form of a discount) so that the overall profit is maximized. They gave a $4$-approximation for the uniform budgets case, and a quasi-PTAS for a special case in which there is an optimal pricing that has only a bounded number of negatively priced items. Here we show that the existence of a quasi-PTAS in the general case is highly unlikely, by showing that the problem is APX-hard.  

In the next section, we give a formal definition of the problem. In Section \ref{s3}, we give a
$O(\log n)$ approximation for trees and in Section \ref{s4} we give a quasi-PTAS for the case of uncrossing paths. We conclude in Section \ref{s5}.
 
\section{The tollbooth problem on trees}\label{s2}
\subsection{Notation} 
Let $\bT=(V,E)$ be a tree. 
We assume that we are given a (multi)set of paths $\cI=\{I_1,\ldots,I_m\}$, defined on the set of edges $E$, where $I_j=[s_j,t_j]\subseteq E$ is the path connecting $s_j$ and $t_j$ in $\bT$. For $I_j\in\cI$, we denote by $B(I_j)\in\RR_+$ the {\it budget} of path $I_j$, i.e., the maximum amount of money customer $j$ is willing to pay for purchasing path $I_j$. In the \emph{tollbooth problem}, denoted henceforth by $\TB$, the objective is to assign a price $p(e)\in\RR_+$ for each edge $e\in E$, and to find a subset $\cJ\subseteq\cI$, so as to maximize 
\begin{equation}\label{obj}
\sum_{I\in\cJ}p(I)
\end{equation} 
subject to the budget constraints
\begin{equation}\label{bud-cons}
p(I)\le B(I),~~\mbox{for all } I\in \cJ,
\end{equation}
where, for $I\in\cI$, $p(I)=\sum_{e\in I}p(e)$. 

\hide{When the number of copies of each edge $e\in E$, available for purchase, is limited by a given number, the problem will be called the \emph{capacitated tollbooth problem}: 
Find a pricing $p:E\mapsto\RR_+$ and a subset $\cJ\subseteq\cI$, so as to maximize \raf{obj} subject to \raf{bud-cons} and the capacity constraints: 
\begin{equation}\label{cap-cons}
|\{I\in\cJ:e\in I\}|\leq c(e),~~\mbox{for all } e\in E.
\end{equation}
}
\medskip

For a node $w\in V$, let $\cI[w]\subseteq\cI$ be the set of paths that pass through $w$. In section \ref{s4}, we will assume that the tree is rooted at some node $\br\in V$.
The depth of $\bT$, denoted $d(\bT)$, is the length of the longest path from the root $\br$ to a leaf. For a node $w\in V$, we denote by $\bT(w)$, the subtree of $\bT$ rooted at $w$ (excluding the path from the parent of $w$ to $\br$), and for a subtree $\bT'$ of $\bT$ we denote by $V(\bT'), E(\bT')$ and $\cI(\bT')$ the vertex set, edge set, and set of intervals contained completely in $\bT',$ respectively. 

\medskip

\subsection{Preliminaries}
In the following sections, we denote by $p^*:E\mapsto\RR_+$ an optimal set of prices, and by $\OPT\subseteq\cI$ the set of intervals purchased in this optimum solution. For a subset of intervals $\cI'\subseteq\cI$, and a price function $p:E\mapsto\RR_+$, we denote by $p(\cI')=\sum_{I\in\cI'}p(I)$ the total price of intervals in $\cI'$. 

\medskip

It easy to see that $n\leq 3m$ may be assumed without loss of generality. 
Indeed, if we root the tree at some vertex $\br$, then for every vertex $v\in V$, we may assume that there is either an interval $I\in\cI$ beginning at $v$ or an interval $I\in\cI$ that passes through two different children of $v$; otherwise, every interval through $v$ must contain its parent $u$ (unless $v=\br$ in which case all edges incident to $\br$ can be contracted), and hence we can contract the edge $e=\{u,v\}$ and increase by $p^*(e)$ the prices of each the edges $\{v,v'\}$ for each child $v'$ of $v$. 

Let $\epsilon>0$ be a given constant. 

\begin{proposition}[\cite{ESZ07}]\label{p1}
Let $p^*$ be an optimal solution for a given instance of $\TB$, and $\epsilon>0$ be a given constant. Then there exists a price function $\tilde{p}:E\mapsto\RR_+$ for which
 
(i) $\tilde{p}(e)\in\{0,1,\ldots,P\}$, for every $e\in E$, where $P=nm/\epsilon$,

(ii) $\tilde{p}(I)\leq\frac{B(I)}{1+\epsilon}$, for every $I\in\OPT$, and

(iii) $\tilde{p}(\OPT)\geq (1-2\epsilon) p^*(\OPT)$.
\end{proposition}  
We shall call the set of prices $\tilde{p}$ satisfying the conditions of Proposition \ref{p1}, \emph{$\epsilon$-optimal prices}.

\medskip

We will make use of the following well-known separator result for trees.

\begin{proposition}\label{p2}
Let $T=(V,E)$ be a tree. Then there exists a node $v$ (called separator node) with the following property: Let $s_1,\ldots,s_r$ be the sizes 
of the components obtained by deleting $v$ from $\bT$, then there is a subset $S\subseteq[r]$ such that 
\begin{equation}\label{b}
\lfloor\frac{n}{3}\rfloor\leq\sum_{i\in S}s_i\leq\lceil\frac{2n}{3}\rceil.
\end{equation}
Such a separator can be found in linear time.
\end{proposition}

This gives a recursive partitioning of $\bT$ in the following standard way: Let $v_0$ be a separator vertex in $\bT$ and $T_1,\ldots,T_r$ be the components of $T-v_0$. Recursively, find separator vertices $v_1,\ldots v_r$ in $\bT_1,\ldots,\bT_r$. We say that node $v_0$ has $\level(v_0)=1$, nodes $v_1,\ldots,v_r$ have level 2, and in general if node $v$ is a separator vertex in the subtree $\bT'$ obtained by deleting one-higher level separator vertex $v'$ then $\level(v)=\level(v')+1$. By \raf{b}, the maximum number of levels $k$ in this decomposition is at most $\log_{3/2}n$. We shall denote by $\cN(\bT)$ the set of separator nodes used in the full decomposition of $\bT$.

\hide{
\subsection{Envy-free pricing}
A pricing $p:E\mapsto\RR_+$ is said to be \emph{envy-free} if for all $j\in[m]$, $p(I_j)<B(I_j)$ implies that customer $j$ can purchase
her preferred set $I_j$. This definition makes sense only in the limited supply-case, since in t he unlimited supply case, every pricing is envy-free. 
While it is not clear how to extend our QPTAS to return such an envy-free pricing, we will be able to ensure that it return a pricing 
satisfying the following approximate envy-free condition.

\begin{definition}\label{d0}\emph{($\epsilon$-Envy-free pricings)}
Let $\epsilon>0$ be a constant. A pricing $p:E\mapsto\RR_+$ is said to be $\epsilon$-\emph{envy-free}, if for all $j\in[m]$, $p(I_j)\leq (1-\epsilon)B(I_j)$ implies that customer $j$ will be able purchase the set $I_j$.
\end{definition}
}

\section{An $O(\log n)$ approximation for the tollbooth problem on trees}\label{s3}

In this section, we prove the following theorem.

\begin{theorem}\label{t1}
There is a deterministic $O(\log n)$-approximation algorithm for $\TB$. 
\end{theorem}
The proof goes along the same lines used in  \cite{BB06} to obtain an $O(\log n)$-approximation for the highway problem. The algorithm consists of 3 main steps: Partitioning, ``randomized cut'', and then dynamic programming. We can then derandomize it to obtain a deterministic algorithm.

\medskip

We say that the given set of paths $\cI$ is \emph{rooted}, if all the paths in $\cI$ start at some node $\br$, called the root of $\bT$.
We will also make use of the following theorem.
\begin{theorem}[\cite{GHKKKM05}]\label{t3}
The tollbooth problem on rooted paths can be solved in polynomial time using dynamic programming.
\end{theorem}

\medskip
 
For $i=1,\ldots,k,$ let 
$$\cI(i)=\{I\in\cI:~i \mbox{ is the smallest level of a separator vertex $v\in\cN(\bT)$ contained in } I\}.$$  
Then $\cI=\cup_{i\in[k]}\cI(i)$ and $I\cap J=\emptyset$ for all $I, J \in\cI(i)$ that contain distinct separators at level $i$.
Let $(\OPT,p^*)$ be an optimal solution. Then, $p^*(\OPT)=\sum_{i=1}^kp^*(\OPT\cap\cI(i))$. Thus if we solve $k$ independent problems on each of the sets $\cI(i)$, $i=1,\ldots,k$, and take the solution with maximum revenue, we get a solution of value at least $p^*(\OPT)/k$.
Thus it remains to show the following result.

\begin{theorem}\label{t5}
Let $v$ be a node of $\bT$, and suppose that all the paths in $\cI$ go through $v$. Then a solution $(\cJ,p)$ of expected value $p(\cJ)\geq p^*(\OPT)/8$ can be found in polynomial time. 
\end{theorem}
\proof
Let $v_1,\ldots,v_r$ be the nodes adjacent to $v$. Note that each path $I\in\cI$ can be divided into two sub-paths starting at $v$; we denote them by $I_1$ and $I_2$. We use the following procedure. 
\begin{enumerate}
\item Let $X\subseteq\{v_1,\ldots,v_r\}$ be a subset obtained by picking each $v_i$ randomly and independently with probability $1/2$. 

\item Let $\cI'=\{I_j\in\cI:~j\in\{1,2\},~I_j\mbox{ contains exactly one vertex of } X\}$.

\item Use dynamic programming (cf. Theorem \ref{t3}) to get an optimal solution $(\cJ,p)$ on the instance defined by $\cI'$ and the tree $T'$ with root $v$ and sub-trees rooted at the children in $X$.

\item Extend $p$ with zeros on all the other arcs not in $T'$, and return $(\cJ,p)$.  
\end{enumerate}
Let $(\OPT,p^*)$ be an optimal solution. We now argue that the solution returned by this algorithm has expected revenue of $p^*(\OPT)/8$.
Clearly, for every $I\in\cI$, either $p^*(I_1)\geq p^*(I)/2$ or $p^*(I_2)\geq p^*(I)/2$; let us call this more profitable part by $I_*$. Then $\sum_{I\in\OPT}p^*(I_*)\geq p^*(\OPT)/2$. Let $\OPT'=\{I\in\OPT:~I_*\mbox{ contains exactly}$
$\mbox{ one vertex of } X\}$. Note that with probability at least $1/4$ each $I\in\OPT$ has $I_*$ intersecting the random set $X$ in exactly one vertex. In particular, 
$$\EE[p^*(\OPT')]=\sum_{I\in\OPT}\EE[p^*(I_*)]\geq \frac{1}{4}\sum_{I\in\OPT}p^*(I_*)\geq \frac{1}{8}p^*(\OPT).$$
Since what our procedure returns is at least as profitable as this quantity, the theorem follows.
\qed

The randomized algorithm above can be derandomized using the method of {\em pairwise independence} \cite{LW05,MR95,BB06}. 

\section{Uncrossing paths}\label{s4}
Here we assume that the tree is rooted at some node $\br\in V$, and that paths in $\cI$ have the following \emph{uncrossing} property:
If $I=[s,t]\in\cI$ then $t$ lies on the path $[s,\br]$. This property implies that once paths in $\cI$ meet they cannot diverge.  

In the course of the solution, we shall consider the following generalized version of the problem: Given intervals as above, and also a function $h:\cI\mapsto\RR_+$, find $\cJ\subseteq\cI$ and a pricing $p:E\mapsto\RR_+$, satisfying \raf{bud-cons} and maximizing $\sum_{I\in\cJ}h(I,p)$.

\medskip

Given a price function $p:E\mapsto\RR_+$ and a node $w\in V$, the \emph{accumulative price} at any node $u$ \emph{on the path} $[w,\br]$ with respect to $w$ is defined as  $p([w,u])$. Obviously, this monotonically increases as $u$ moves towards the root. In this section we prove the following theorem.

\begin{theorem}\label{t2}
There is a quasi-polynomial time approximation scheme for the tollbooth problem with uncrossing paths. 
\end{theorem}

In the following, we fix $K=\lceil\log (nP)/\log(1+\epsilon)\rceil$.

\begin{definition}\label{d2}\emph{($\epsilon$-Relative pricings)}
Let $w\in V$ be a given node of $\bT$, and $0\leq k\leq K$ and $0\leq k'\leq 2\log_{3/2} n$ be given integers. 
We call any selection of $k$ nodes $u_1,\ldots,u_{k}\in V$, $k$ indices $-\infty\leq i_1<\cdots<i_{k}\le K$, and $k'$ values $p_1,\ldots,p_{k'}\in\{0,1,\ldots,nP\}$, such that $w,u_1,u_2,\ldots,u_{k},\br$ lie on the path $[w,\br]$ in that order,
an \emph{$\epsilon$-relative pricing} w.r.t. $w$, and denote it by $(w,k,k',u_1,\ldots,u_{k},i_1,\ldots,i_{k},p_1,\ldots,p_{k'})$.
\end{definition}
The total number of possible $\epsilon$-relative pricings with respect to a given $w\in V$ is at most 
\begin{equation}\label{bd}
L=(d(T)K)^{K}(nP+1)^{2\log_{3/2}n},  
\end{equation}
which is $m^{\polylog(m)}$ for every fixed $\epsilon>0$.

\medskip

\begin{definition}\label{d4}\emph{(Consistent pricings)} 
Let $R=(w,k,k',u_1,\ldots,u_{k},i_1,\ldots,i_k,p_1,\ldots,p_{k'})$ be an $\epsilon$-relative pricing w.r.t. node $w\in V$, $\cL=\{s_1,\ldots,s_{k'}\}$ be the set of separators from $\cN(\bT)$ on the path from $(w,\br]$, and $p:E\mapsto\RR_+$ be a pricing of  $E$. We say that $R$ is $\epsilon$-\emph{consistent} with $p$ and $\cL$ if 
\begin{enumerate}
\item[(C1)]  for $j=1,\ldots,k-1$, $(1+\epsilon)^{i_{j}}\leq p([w,u])\leq(1+\epsilon)^{i_{j}+1}$ if $u$ lies in the interval $[u_j,u_{j+1})$ (excluding $u_{j+1}$),
\item[(C2)]  for $j=1,\ldots,k'$, $p([w,s_j])=p_j$.
\end{enumerate}
\end{definition}

\medskip

\begin{lemma}\label{l1}
Let $\tilde{p}:E\mapsto\RR_+$ be an $\epsilon$-optimal pricing for a given instance of $\TB$, $w\in V$ be an arbitrary node, and $\cL=\{s_1,\ldots,s_{k'}\}$ be the set of separators in $\cN(\bT)$ on the path from $(w,\br]$. Then there exists an $\epsilon$-relative pricing $R$ w.r.t. $w$, that is $\epsilon$-consistent with $\tilde{p}$ and $\cL$. 
\end{lemma}  

\medskip

With every $\epsilon$-relative pricing $R$, we can associate a system of linear inequalities, denoted by $S(R)$, on a set of $E$ variables $\{p(e):~e\in E\}$, consisting of the constraints $(C1)$ and $(C2)$, together with the non-negativity constraints $p(e)\geq 0$. The feasible set for this system gives the set of all possible pricings with which $R$ is $\epsilon$-consistent. For two systems of inequalities $S_1,S_2$, we denote by $S_1\wedge S_2$ the system obtained by combining their inequalities.

\medskip

Let $R=(w,k,k',u_1,\ldots,u_{k},i_1,\ldots,i_k,i'_1,\ldots,i'_{k'})$ be an $\epsilon$-relative pricing w.r.t. a node $w\in V$. Given an interval $I\in\cI[w]$, we associate a value $v(I,R)$ to $I$, defined with respect to $R$ as follows: Let $j(I)$ be the largest index such that $u_{i_{j(I)}}$ is contained in $I$.  
Then, define 
$
v(I,R)=
(1+\epsilon)^{j(I)}.
$
For a subset of intervals $\cI'\subseteq\cI$, we define, as usual, $v(\cI',R)=\sum_{I\in\cI'}v(I,R)$. 
It follows that for any $\epsilon$-relative pricing $R$ w.r.t. a node $w\in V$, any $p:E\mapsto\RR_+$ with which $R$ is consistent, and any $I=[s,t]\in\cI[w]$, we have
\begin{equation}\label{ext}
v(I,R)\leq p([w,t])\leq (1+\epsilon)v(I,R). 
\end{equation}

\medskip

\noindent{\bf Decomposition into two subproblems.}~ Let $w\in\cN(\bT)$ be a separator node. Then $\bT$ can be decomposed into two subtrees
$\bT_L=(V_L,E_L)$ and $\bT_R=(V_R,E_R)$, such that the root $\br\in V_R$ and $w$ is the root of $\bT_L$. 
We define two $\TB$ instances $(\bT_L,\cI_L)$ and $(\bT_R,\cI_R)$ where:
\begin{eqnarray*}
\cI_0&=&\{[s,t]\in\cI[w]~:~s\in V_L \mbox{ and }t\in V_R\},\\
\cI_L&=&\{[s,t]\in\cI~:~s,t\in V_L\}\cup\{[s,w]~:~[s,t]\in\cI_0\},\\
\cI_R&=&\{[s,t]\in\cI~:~s,t\in V_R\}.
\end{eqnarray*}
In other words, the intervals passing through $w$, crossing from $\bT_L$ to $\bT_R$ are truncated in $\bT_L$ while all other intervals remain the same\footnote{throughout, we will make the implicit assumption that each interval has an "identity"; so, for instance, $\cI_L\cap \cI_0$ will be used to denote the set $\{I\in\cI_0:~I=[s,t]\mbox{ and }[s,w]\in\cI_L\}$}. Note that from the choice of $w$, we have $\max\{|V(\bT_L)|,|V(\bT_R)|\}\le \frac{2n}{3}+1$, and both instances $(\bT_L,\cI_L)$ and $(\bT_R,\cI_R)$ are of the uncrossing type, with roots $w$ and $\br$, respectively.   

\medskip

The algorithm is shown in Figure \ref{f1}. It is initially called with an empty $\cS$, and with $h(I)=0$ for all $I\in\cI$. The procedure iterates over all $\epsilon$-relative pricings $R$, consistent with $\cS$, w.r.t. the middle edge $e^*,$ then recurses on the subsets of intervals to the left and right of $e^*$. Intervals crossing from $\bT_L$ to $\bT_R$ will be truncated 
and their values will be charged to $\bT_L$; hence the corresponding budgets are reduced, and the corresponding $h$-values are increased. 

\medskip

\noindent {\bf Solving the base case}.~ At the lowest level of recursion (either line 1 or 4), we have to solve a linear program defined by the system $\cS$. Note that the system may contain constraints on variables outside the current set of edges $E$ of the current tree $\bT$ (resulting from previous nodes of the recursion tree). However, we can reduce this LP to one that involves only variables in $E$. Indeed, any constraint that involves a variable not in $E$, has the form $L\leq p([w,u])\leq U$, where $u\in V(\bT)$, and $w\not\in V(\bT)$ is a separator node such that there is another separator node $w'\in V(\bT)$ on the path from $w$ to $u$. Then when $w'$ was considered in the recursion, a constraint of the form $p([w,w')]=q$, for some value $q$, was appended to $\cS$ (recall $(C2)$ in the definition of consistent pricings). Now, we can replace the first constraint by the equivalent constraint $L-q\leq p([w',u])\leq U-q,$ which only involves variables from $E$. This is exactly what procedure REDUCE$(\cS,\cdot)$ does in lines 2 and 6.

\medskip

When the procedure returns, we get a pricing $p:E\mapsto\RR_+$ and a set of intervals $\cJ\subseteq\cI$ which can be purchased under this pricing.

\begin{figure}
{\small
\begin{code}
\noindent{\bf Algorithm $\TB(\bT,\cI,\br,B,h,\cS)$:}\\
{\it Input:} An uncrossing $\TB$ instance $(\bT=(V,E),\cI)$ with root $\br$, \\
\>\> budgets and values $B,h:\cI\mapsto\RR_+$, and a feasible system of inequalities $\cS$\\
{\it Output:} A pricing $p:E\mapsto\RR_+$ and a subset $\cJ\subseteq\cI$\\\\

1.\>\> \If $|\cI|=0$, \Then \\
2.\>\>\> $\cS'\leftarrow$REDUCE$(\cS,E)$\\
3.\>\>\> \Return $(p,\emptyset)$, where $p$ is any feasible solution of $\cS'$\\
4.\>\> \If $d(\bT)=1$, \Then\\
5.\>\>\> \Foreach edge $e$ of $\bT$ \Do\\
6.\>\>\>\> $\cS'\leftarrow$REDUCE$(\cS,\{e\})$\\
7.\>\>\>\> $p(e)\leftarrow\argmax\{\sum_{I\in\cI:~p'\leq B(I)}(h(I)+p')~:~p'\mbox{ satisfies }\cS'\}$\\
8.\>\>\>\> $\cJ(e)\leftarrow\{I\in\cI:B(I)\geq p(e)\}$\\
9.\>\>\> \Return $((p(e):e\in E),\bigcup_{e\in E}\cJ(e))$\\     
10.\>\> let $w$ be a separator node of $\bT$ and $\bT_L,\bT_R,\cI_0,\cI_L,\cI_R$ be as defined above\\
11.\>\> \For every $\epsilon$-relative pricing $R$ w.r.t. $w$ for which $S\wedge S(R)$ is feasible \Do\\
12.\>\>\> \Foreach $I\in\cI_0$ \Do\\
13.\>\>\>\> $B(I)\leftarrow B(I)-(1+\epsilon)v(I,R)$\\
14.\>\>\>\> $h(I)\leftarrow h(I)+v(I,R)$\\ 
15.\>\>\> $(p_1,\cJ_1)\leftarrow$ $\TB(\bT_L,\cI_L,w,B,h,\cS)$\\
16.\>\>\> $(p_2,\cJ_2)\leftarrow$ $\TB(\bT_R,\cI_R,\br,B,h,\cS\wedge S(R))$\\
17.\>\>\> let $p$ be the pricing defined by $p(e)=p_1(e)$ if $e\in E_L$ and $p(e)=p_1(e)$ if $e\in E_R$ \\
18.\>\>\> $\cJ\leftarrow\cJ_1\cup\cJ_2$\\
19.\>\>\> record $(p,\cJ)$\\
20.\>\>   \Return the recorded solution with largest $p(\cJ)+h(\cJ)$ value\\
\end{code}
}

\caption{The procedure for computing $\epsilon$-approximate prices.}
\label{f1}
\end{figure}

Theorem \ref{t2} follows from the following two lemmas.

\begin{lemma}\label{l2}
Algorithm $\TB$ runs in quasi-polynomial time in $m$, for any fixed $\epsilon>0$.
\end{lemma}

\begin{lemma}\label{l3}
For any $\epsilon>0$, Algorithm $\TB$ returns a pricing $p$ and a set of intervals $\cJ$ such that $p(I)\le B(I)$ for all $I\in\cJ$ and $p(\cJ)\geq (1-3\epsilon)p^*(\OPT)$.
\end{lemma}

\section{Hardness of the highway problem}\label{HARD}
\subsection{Strong NP-hardness in the standard model}\label{StrongNP}
Recall that the highway problem is the special case of the tollbooth problem when the underlying graph is a path.
In \cite{GHKKKM05}, Guruswami, et al. considered the highway problem and gave a polynomial time
algorithm when the maximum budget is bounded by a constant, and all the budgets are integral. 
Balcan and Blum \cite{BB06} gave a constant factor
approximation algorithm when all intervals have the same length. 
Breist and Krysta \cite{BK06} showed that the problem is weakly NP-hard.
In \cite{GLSUT08}, Grigoriev et al. showed that
a restricted version of the problem when the prices are required to satisfy a {\em monotonicity} condition
remains weakly NP-hard. 
In this section, we show that the problem is strongly NP-hard by a reduction
from MAX-2-SAT.

Consider a MAX-2-SAT instance
with $n$ variables $\{x_1, \cdots, x_n\}$ and $m$ clauses $\{C_1, \cdots, C_m\}$.
Let the variables be numbered $1, \cdots, n$. We construct a gadget for each variable
and each clause. We start by describing the gadgets in our construction.

\subsubsection{Variable Gadget}
The variable gadget for each variable consists of two copies of the following {\em basic gadget} and
a {\em consistency gadget}. We first describe the basic gadget, and then describe the consistency
gadget and the construction of a variable gadget.

\noindent
{\em Basic Gadget:} The basic gadget consists of $4$ edges $e_1, \cdots, e_4$, and  
$4$ types of intervals $A,B,C$ and $D$. There are $4$ intervals each of type $A$ and
$B$, labeled $a_1, \cdots, a_4$, and $b_1, \cdots, b_4$ respectively. The intervals
$a_i = b_i = [e_i]$, $i=1, \cdots 4$. The intervals $a_1, \cdots,
a_4$ have budgets of $1,2,2,1$ respectively, and the intervals $b_1, \cdots, b_4$ have budgets
$2,1,1,2$ respectively. There are $2$ type $C$ intervals, $c_1$ and $c_2$, 
with $c_1 = [e_1, e_2]$, and $c_2 = [e_3, e_4]$. 
These intervals have a budget of $3$. There are two intervals of type $D$, $d_1 = d_2 = [e_2,e_3]$
with 
$d_1$ having a budget of $4$, and $d_2$, a budget of $2$.
The basic gadget is shown in Figure \ref{fig:pricevargad}. We now show that there are exactly two price
assignments for $\{e_1, \cdots, e_4\}$ that gives us optimum profit.

\begin{figure}[htpb]
\begin{center}
\input{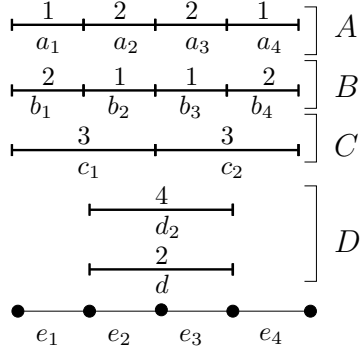}
\end{center}
\caption{A basic gadget. The gadget consists of 4 edges, and 4 types of intervals $A, B, C$ and $D$. The 
interval labels are shown below each interval, and the budgets are shown above each interval.}
\label{fig:pricevargad}
\end{figure}

\begin{lemma}
\label{lem:basegadget}
The maximum profit that can be obtained from a basic gadget is $18$, and there are exactly
two sets of prices that achieve this profit.
\end{lemma}

We call the price assignment $(1,2,2,1)$ to the edges $e_1, \cdots, e_4$ respectively,  a 
TRUE assignment, 
and the price assignment $(2,1,1,2)$ to the edges $e_1, \cdots, e_4$ respectively, 
a FALSE assignment.
The variable gadget is constructed on $8n+1$ edges
$(e_{4n}, e_{4n-1}, \cdots, e_{1}, h, f_{1}, \cdots, f_{4n})$, where $n$ is the number of variables in the MAX-2-SAT instance. 
Each variable gadget consists of two copies of the basic gadget, along with a consistency gadget.
The consistency gadget ensures that the two basic gadgets have the same price assignment, i.e.,
both set to TRUE, or both set to FALSE. 
 More formally, let $(x_1,\cdots, x_n)$ be an order on the variables
of the MAX-2-SAT instance. Then, the gadget for variable $x_i$, consists of two basic gadgets,
$B^1_i$ and $B^2_i$. $B^1_i$ consists of intervals (customers) interested in the edges $e_{4i - 3}, \cdots, e_{4i}$
and $B^2_i$ consists of intervals interested in the edges $f_{4i-3}, \cdots, f_{4i}$.
Finally, the intervals ensuring consistency of the gadget for variable $x_i$
spans from $e_{4i-1}, \cdots, f_{4i-3}$. The consistency gadget consists of a single interval
that has a budget of $mn^2 + 6(2i-2) + 6$. 
Finally, we add a new type of interval, called a type $H$ interval
that is interested only in the edge $h$, and has a budget of $mn^2$. 

Figure \ref{fig:pricevarfull} shows the arrangement of the variable gadgets. We now show that the
consistency intervals do their job. i.e., if for a variable gadget, $B^1_i$ and $B^2_i$ have different
price assignments, we obtain a smaller profit than when they are the same. 

\begin{figure}[htpb]
\begin{center}
\input{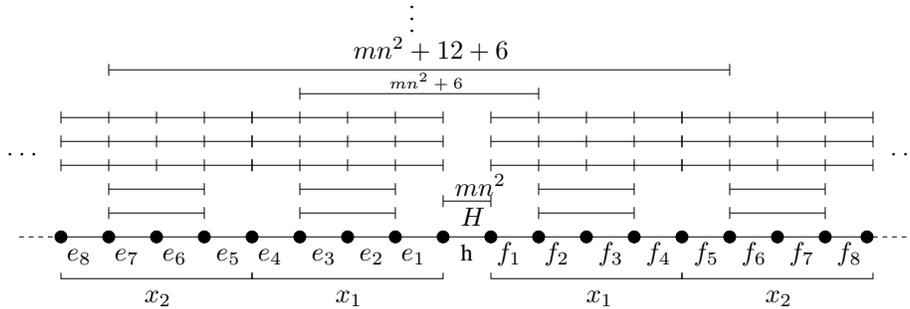}
\end{center}
\caption{The variable gadget.}
\label{fig:pricevarfull}
\end{figure}

\begin{lemma}
\label{lem:consistency}
The maximum profit of $2mn^2 + 6(2i-2) + 6 + 36$ from a variable gadget and 
the interval $h$ is achieved only
when both the basic gadgets corresponding to a variable are consistent, and the type $H$
interval purchases edge $h$ at a price of $mn^2$. 
\end{lemma}

We will create several copies of the basic gadgets, the consistency gadgets for each variable as well
as several copies of the $H$ interval to ensure
that in an optimum price assignment, the basic gadgets are consistent, and the reduction goes
through. But before we do this, we describe the clause gadgets. 

\subsubsection{Clause Gadgets}
The clause gadget for a clause of variables $x_i$ and $x_j$ runs between the basic gadget $B^1_i$
and $B^2_j$. There are four types of clause gadgets corresponding to the four types of clauses.
Each clause gadget consists of one interval. These intervals have the property that 
we obtain a certain revenue from the clause interval if and only if the clause is satisfied; otherwise
we obtain nothing.  The clause gadgets for the four types of clauses are shown in 
Table \ref{fig:clausetab} and in 
Figures \ref{fig:prxy}, \ref{fig:prxcy}, \ref{fig:prxyc}, and \ref{fig:prxcyc} in the Appendix.

\begin{figure}
\begin{center}
\begin{tabular}{|c|c|c|}
\hline
Clause & Interval & Budget \\\hline
\xy & $[e_{4i-3}, f_{4j-3}]$ & $mn^2 + 6(i + j - 2) + 3$\\ \hline
\xcy & $[e_{4i-1},  f_{4j-3}]$ & $mn^2 + 6(i + j - 2) + 6$\\ \hline
\xyc & $[e_{4i-3}, f_{4j-1}]$ & $mn^2 + 6(i + j - 2) + 6$ \\ \hline
\xcyc &$[e_{4i-1}, f_{4j-1}]$ & $mn^2 + 6(i + j - 2) + 9$ \\ \hline
\end{tabular}
\end{center}
\caption{This table shows the lengths and budgets of the intervals making up a clause gadget for the four different kinds of clauses.}
\label{fig:clausetab}
\end{figure}

We say that a pricing is {\em consistent} if for every variable, 
the price assignment to the two basic gadgets of the variable gadget are both TRUE or both FALSE, and
the consistency intervals spend their entire budgets.

\begin{lemma}
\label{lem:clause}
Consider a clause $C$ consisting of variables $x_i$ and $x_j$ and a consistent price assignment
to the edges. Then, the intervals corresponding to $C$ will be able to purchase their desired edges
if and only if the corresponding truth assignment to the variables satisfies the clause $C$.
\end{lemma}

\subsubsection{NP-hardness}
We now describe the final reduction. As mentioned earlier, we have to create copies of the variable gadget, 
consistency gadget and the $H$ interval for the proof to go through. We make $T$ copies
of each basic gadget, of each consistency gadget, and of the $H$ interval, where any value of
$T$, larger than $m^2n^3$ will suffice for the proof.
Observe that for a variable gadget again, the profit maximizing prices achieve
consistency of the variable gadget, and making $T$ copies of the $H$ intervals ensures that
the price of the edge $h$ is set to $mn^2$.

\begin{theorem}\label{t-strongNP}
The highway problem is strongly NP-hard.
\end{theorem}

\subsection{APX-hardness in the discount model}
 
\begin{theorem}
The highway problem with negative prices is APX-hard, even restricted to instances in which one edge is shared by all customers.
\end{theorem}
\begin{proof}
We will show that the problem is equivalent to a pricing problem on bipartite graphs and prove that the latter problem is APX-hard. 
Assume we are given an instance of the highway problem in which edge $e$ is contained in each of the intervals. We split edge $e$ by adding a node $v_0$ on $e$. This has no effect on the problem. We construct a bipartite graph $H$ with one set consisting of the points left of $v_0$ and the other set consisting of the points right of $v_0$. Now, an interval containing $e$ becomes an edge in the bipartite graph.  The items are the vertices and a customer is interested in the two items on the vertices of its edge. Given a pricing $p$ for the highway instance, we define the pricing $q_p$ of $H$ by letting the price of a vertex $v$ in $H$ be the cumulative price $p([v_0,v])$. Conversely, for any pricing $q$ of the vertices of $H$ there is a corresponding  pricing $p$ of the highway problem such that $q=q_p$.
\begin{figure}
  \center\includegraphics[width=11cm]{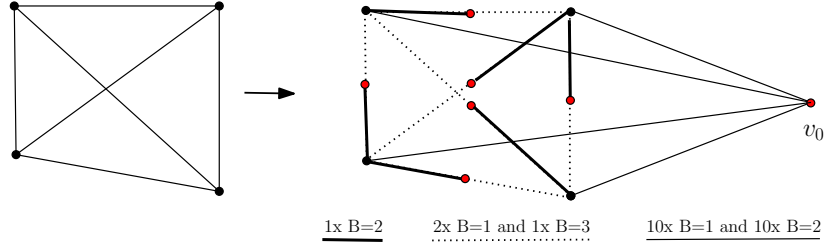}\\
  \caption{Reduction: Each edges $e$ is partitioned in edges $e_1$ (fat) and $e_2$ (dotted). Which is which is arbitrary. The notation 1x B=2 means that there is one customer with budget 2 on this type of edge.}\label{fig:APX}
\end{figure}

We prove that the pricing problem on bipartite graphs is APX-hard by a reduction from maxcut on 3-regular graphs. Given a 3-regular graph  $G=(V,E)$, we make it bipartite by placing an extra vertex $v_e$ on every edge $e$, dividing it in two new edges $e_1$ and $e_2$. For $e_1$ there is one customer with budget $2$ and for $e_2$ there are two customers with budget $1$  and one customer with budget $3$.  We define one extra vertex $v_0$ and define for each $v\in V$ an edge $a_v=(v,v_0)$.  For each such edge $a_v$ there are ten customers with budget $1$ and ten customers with budget $2$. The bipartite graph partitions into $V$ and the new vertices $V'=\{v_e\mid e\in E\}\cup\{v_0\}$. (To enhance reading we write prices and budgets in decimal and amounts of customers in words.)

Consider any pricing $p$ of  the bipartite graph. We may assume that  $p(v_0)=0$ since  subtracting $p(v_0)$ from all vertices in $V'$ and adding   $p(v_0)$ to all vertices in $V$ does not change the profit. We will have to take into account though that prices may be negative.

Next, we argue that in any optimal solution the price $p(v)$ of any vertex $v\in V$ is either $1$ or $2$. Denote by $p(e)$ the profit we get from the  customer on $e_1$ plus the three customers on $e_2$. It is easy to see that $3\le p(e)\le 5$ for any edge $e$. Suppose $p(v)=\alpha\notin\{1,2\}$.
If $\alpha>2$ or $\alpha\le 0$ then the profit on $a_v$ is 0. By changing the price to $p(v)=2$ the profit becomes ten times $2$ is 20. The maximum profit on any of the three adjacent edges $e$ is $5$. Hence, we gain $20$ and loose at most $15$.
Now assume $0<\alpha<1$. We raise the price of $v$ to price $p(v)=1$ and reduce the price on the vertices $v_e$ by $1-\alpha$ for each of the three adjacent edges $e$ of $v$ in $G$. The 20 customers on $a_v$ add an extra $20(1-\alpha)$ to the profit. No other customer sees an increase of its bundle price and at most 9 customers will see a reduction of the bundle price. Hence, we loose at most $9(1-\alpha)$ on them. Now assume $1<\alpha<2$.
We raise the price to $p(v)=2$ and reduce the price on the vertices $v_e$ of adjacent edges $e$ by $2-\alpha$. The argument is the same: We gain $10(2-\alpha)$ and loose at most $9(2-\alpha)$ since at most $9$ customers will see their bundle price drop.

We showed that there is an optimal pricing in which $p(v_0)=0$ and $p(v)\in\{1,2\}$ for all $v\in V$. Next we prove that there is a cut of size $k$ in $G$ if and only if the maximum profit is $20|V|+4|E|+k$. Given a cut of size $k$ we price the vertices on one side 1 and on the other side 2. Now consider an edge $e=(v,w)\in E$ with $p(v)=p(w)=1$. We can get a profit $p(e)=4$ by setting  $p(v_e)=0$ or $p(v_e)=1$. This is also the maximum profit possible. Similarly, if $p(v)=p(w)=2$ then we can get the maximum profit $p(e)=4$ by setting $p(v_e)=0$. Finally, if $p(v)=1$ and $p(w)=2$ then we can get the maximum possible profit $p(e)=5$ by setting $p(v_e)=0$ or $p(v_e)=1$, depending on how we chose $e_1$ and $e_2$. From the customers on edges adjacent to $v_0$ we get the maximum profit of 20 per edge.
The total profit is exactly  $20|V|+5k+4(|E|-k)=20n+4|E|+k$ and this is maximum possible if the maximum cut is $k$. The reduction is gap-preserving since $3|V|=2|E|$ and $|E|\le 2k$.\qed
\end{proof}

In the reduction we showed that there always is an optimal pricing of the bipartite graph with only non-negative prices. Hence, the bipartite graph pricing problem remains APX-hard if we restrict to non-negative prices.
\begin{corollary}
The graph pricing problem is APX-hard on bipartite graphs and all budgets in $\{1,2,3\}$. This holds for the non-negative version as well as for the version with negative prices allowed.
\end{corollary}

Guruswami et al.~\cite{GHKKKM05} show that the graph-pricing problem is APX-hard even if all budgets are 1. Note that the bipartite case is trivially solved in that case by setting a price of 1 to all items on one side.

\section{Conclusion}\label{s5}
In this paper, we presented an $O(\log n)$-approximation algorithm for the tollbooth problem on trees, which is better than
the upper bound currently known for the general problem. Improving this bound is an interesting open problem. One plausible direction
towards this is to use as a subroutine, the quasi-polynomial time algorithm for the case of uncrossing paths. 
Such techniques have been used before, for example for the multicut problem on trees \cite{GNM06}. 
However,
it is
unclear how a general instance of the \TB\ problem can be decomposed into a set of problems of the uncrossing type.
For the highway
problem, the strong NP-hardness presented in this paper 
shows that the problem is almost closed, modulo improving the running time
from quasi-polynomial to polynomial. 

\medskip

\noindent
{\bf Acknowledgements:} We would like to thank Naveen Garg for suggesting to use the separator theorem in the proof of Theorem \ref{t1}, and Chaitanya Swamy for helpful remarks.

\bibliographystyle{amsplain}
\bibliography{tollbooth}

\newpage
\noindent
{\bf\Large Appendix A: Proofs}

\medskip

\noindent{\bf Proof of Lemma \ref{l1}}.~
Let $R=(w,k,k',u_1,\ldots,u_{k},i_1,\ldots,i_k,p_1,\ldots,p_{k'})$ be defined as follows: write $u_0=w$ and $i_0=-\infty$, and let $i_j$ and $u_j$, for $j=1,2,\ldots$, be respectively the smallest non-negative index and the closest node to $u_{j-1}$ on the path $[u_{j-1},\br]$ with $(1+\epsilon)^{i_{j-1}}<\tilde{p}([w,u_{j}])\leq (1+\epsilon)^{i_{j}}$; $k$ will be the largest such index $j$.
Finally, for $j=1,\ldots,k'$, let $p_j=p([w,s_j])$. Note that $k\le K$ since $(1+\epsilon)^k \leq nP$, and $k'\le 2\log_{3/2} n$ since the number of separators on the path from any node to the root is at most $2\log_{3/2} n$. 
\qed

\medskip

\noindent{\bf Proof of Lemma \ref{l2}}.~
The number of possible $\epsilon$-relative pricing is at most $L$, given in \raf{bd}. This gives the recurrence
$$
T(n)\leq \poly(n,m)+2L \cdot T(\frac{2n}{3}).
$$
for the running time. Thus $T(n)\leq L^{O(\log n)}\poly(m)$ and the lemma follows.
\qed

\medskip

\noindent{\bf Proof of Lemma \ref{l3}}.~ Let $(p_{\bT},\cJ_{\bT})$ be the solution returned by the algorithm when the input is $(\bT,\cI_{\bT},\br_{\bT},B_{\bT},h_{\bT},\cS_{\bT})$.  
We show by induction on the depth of the recursion tree that, if there exists a pricing $p'_{\bT}$ satisfying $\cS_{\bT}$, and a set $\cJ'_{\bT}$ such that $p'_{\bT}(I)\leq B_{\bT}(I)/(1+\epsilon)$ for all $I\in\cJ'_{\bT}$, then
\begin{itemize}
\item[(i)] $p_{\bT}(I)\le B_{\bT}(I)$ for all $I\in\cJ_{\bT}$, and $p_{\bT}$ is feasible for $\cS_{\bT}$; and
\item[(ii)] $p_{\bT}(\cJ_{\bT})+h_{\bT}(\cJ_{\bT})\geq \frac{p'_{\bT}(\cJ'_{\bT})}{1+\epsilon}+h_{\bT}(\cJ'_{\bT})$.
\end{itemize}

The statement of the theorem follows from (i) and (ii) by taking, at the highest level where $h_{\bT}(I)=0$ for all $I$,  $p'_{\bT}=\tilde{p}$ (an $\epsilon$-optimal pricing) and $\cJ'_{\bT}=\OPT$. 

\noindent{\it Base case.}~ At a leaf of the recursion tree, we either have $|\cI|=0$ in which case (i) and (ii) are trivially satisfied, 
or $d(\bT)=1$ in which case (i) and the stronger version of (ii), $p_{\bT}(\cJ_{\bT})+h(\cJ_{\bT})\geq p'_{\bT}(\cJ'_{\bT})+h_{\bT}(\cJ'_{\bT})$, are insured by 
the computation in line 7. 

\noindent{\it General recursion level.}~ Let $w,\bT_L,\bT_R,\cI_0,\cI_L,\cI_R$ be as defined in line 10 at the current level, and $p_1,p_2,\cJ_1$ and $\cJ_2$ the returned pricings and sets at lines 15 and 16. Let $R_{\bT}$ be an $\epsilon$-relative pricing consistent with $p'_{\bT}$. 
Then the restrictions $p'_{\bT_L}$ and $p'_{\bT_R}$ of $p'_{\bT}$ on $\bT_L$ and $\bT_R$ satisfy, respectively, $\cS_{\bT_L}$ and $\cS_{\bT_R}$. 
Moreover, for any $I\in\cJ'_{\bT}\cap(\cI_L\setminus \cI_0)$, we have $p_{\bT_L}'(I)\leq \frac{B_{\bT_L}(I)}{1+\epsilon}$; for any $I\in\cJ'_{\bT}\cap\cI_R$, we have $p_{\bT_R}'(I)\leq \frac{B_{\bT_R}(I)}{1+\epsilon}$; and for any $I=[s,t]\in\cJ'_{\bT}\cap\cI_0$, we have $p_{\bT_L}'(I)=p_{\bT}'([s,w])=p_{\bT}'([s,t])-p_{\bT}'([w,t])\leq p_{\bT}'([s,t])-v(I,R_{\bT})\leq \frac{B_{\bT}(I)}{1+\epsilon}-v(I,R_{\bT})=\frac{B_{\bT_L}(I)}{1+\epsilon},$ where the first inequality follow from \raf{ext}, and the last equation follows from line 13 of the procedure. Thus we can apply the induction hypothesis to the two subproblems, and hence get that
\begin{eqnarray}
\label{e.1.1}
p_1(I)&\le& B_{\bT_L}(I) \mbox{ for all }I\in \cJ_1,\\
\label{e.1.2}
p_2(I)&\le& B_{\bT_R}(I)  \mbox{ for all }I\in \cJ_2,\\
\label{e.1.3}
p_1(\cJ_1)+h_{\bT_L}(\cJ_1)&\geq& \frac{p'_{\bT_L}(\cJ'_{\bT}\cap\cI_L)}{1+\epsilon}+h_{\bT_L}(\cJ'_{\bT}\cap\cI_L),\\
\label{e.1.4}
p_2(\cJ_2)+h_{\bT_R}(\cJ_2)&\geq& \frac{p'_{\bT_R}(\cJ'_{\bT}\cap\cI_R)}{1+\epsilon}+h_{\bT_R}(\cJ'_{\bT}\cap\cI_R),
\end{eqnarray}
and both $p_1$ and $p_2$, and hence $p_{\bT}$, satisfy $\cS_{\bT}$.
By \raf{e.1.1} and \raf{e.1.2}, we have $p_{\bT}(I)\leq B_{\bT}(I)$ for all $I\in\cJ_1\cup \cJ_2\setminus \cI_0$. By \raf{e.1.1} and line 13 of the procedure, we also have $p_1([s,w])\le B_{\bT}(I)-(1+\epsilon)v(I,R_{\bT})$ for all $I=[s,t] \in\cI_0\cap\cJ_1$. Since $p_2$ satisfies $\cS(R_{\bT})$ (c.f. line 16), and hence is $\epsilon$-consistent with $R_{\bT}$, we get by \raf{ext} that $p_2([w,t])\le (1+\epsilon)v(I,R_{\bT})$ for all $I=[s,t]\in\cI_0$. Combining this with the above inequality gives $p_{\bT}(I)=p_1[s,w]+p_2[w,t]\le B_{\bT}(I)$ for all $I=[s,t]\in\cI_0\cap\cJ_1$, and hence proves (i).

Now we prove (ii). We have the following: $p'_{\bT_L}(\cJ'_{\bT}\cap\cI_L\setminus\cI_0)=p'_{\bT}(\cJ_T'\cap\cI_L\setminus\cI_0)$, $h_{\bT_L}(\cJ_T'\cap\cI_L\setminus\cI_0)=h_{\bT}(\cJ_T'\cap\cI_L\setminus\cI_0)$, $p'_{\bT_R}(\cJ'_{\bT}\cap I_R)=p'_{\bT}(\cJ_T'\cap\cI_R)$, $h_{\bT_R}(\cJ'_{\bT}\cap I_R)=h_{\bT}(\cJ_T'\cap\cI_R)$, and
\begin{eqnarray*}
p'_{\bT_L}(\cJ_{\bT}'\cap\cI_0)&=&\sum_{I=[s,t]\in\cJ_{\bT}'\cap\cI_0}p'_{\bT}([s,w]),\\
h_{\bT_L}(\cJ_{\bT}'\cap\cI_0)&=&\sum_{I=[s,t]\in\cJ_{\bT}'\cap\cI_0}h_{\bT}([s,t])+\sum_{I=[s,t]\in\cJ_{\bT}'\cap\cI_0}v(I,R_{\bT}),\\
                              &\geq& \sum_{I\in\cJ_{\bT}'\cap\cI_0}h_{\bT}(I)+\sum_{I=[s,t]\in\cJ_{\bT}'\cap\cI_0}\frac{p_{\bT}'([w,t])}{1+\epsilon},
\end{eqnarray*}
where the last inequality follows by \raf{ext}. Summing all these together gives (ii) and concludes the proof of the lemma.
\qed

\medskip

\noindent{\bf Proof of Lemma \ref{lem:basegadget}}.~ 
Consider the pair of intervals $\{a_i,b_i\}$ for each $i = 1, \cdots, 4$. The maximum
profit that can be obtained from such a pair is $2$, which is obtained by setting
either $p(e_i) = 1$, or $p(e_i) = 2$. Any other price clearly yields a smaller
profit. 
Similarly if we consider only the intervals of type $D$,
the maximum profit is obtained by setting $p(e_2) + p(e_3) = 2$, or 
$p(e_2) + p(e_3) = 4$. This gives us $4$ price vectors that give us maximum
profit from all except the type $C$ intervals, viz.
$(1,1,1,1), (2,2,2,2), (1,2,2,1), (2,1,1,2)$. In the first case, we only
obtain a profit of $4$ from the type $C$ intervals for a total profit of
$16$, while in the second case, we exceed the budget
of both the type $C$ intervals giving us a profit of only $12$.
Thus there are only two profit maximizing price assignments.\qed
\qed

\medskip

\noindent{\bf Proof of Lemma \ref{lem:consistency}}.~ 
Consider the gadget for variable $x_i$. If the gadget is consistent, we see that both the consistency
gadget, and the type $H$ interval spend their entire budget, 
and we obtain a profit of $2mn^2 + 6(2i-2) + 6 + 36$. 
Suppose $B^1_i$ is TRUE and $B^2_i$ is FALSE. Then, we are forced to set the price of
edge $h$ to $mn^2-1$, otherwise the consistency gadget is unable to purchase it's edges and we lose
at least $mn^2 + 6(2i-2) + 6$ from the total profit. However, by setting $p(h) = mn^2-1$, the maximum profit
we obtain is at most $(mn^2-1) + mn^2 + 6(2i-2) + 36$, which is smaller than the maximum profit by $1$ unit. 
On the other hand, if $B^1_i$ is FALSE, and $B^2_i$ is TRUE, we lose $1$ unit from the maximum profit
since we cannot raise the price of edge $h$ to more than $mn^2$, and the
consistency gadget is unable to spend it's entire budget. Hence, the maximum profit is obtained only when the variable gadget
is consistent.
\qed

\medskip

\noindent{\bf Proof of Lemma \ref{lem:clause}}.~ 
Consider a consistent price assignment, with the edge $h$ having a price of $n^2$ and a clause \xy . 
If the clause \xy\ is 
not satisfied, then the gadgets for variables $x_i$ and $x_j$ have a FALSE price assignment,
and the prices for the edges in the gadgets for $x_i$ and $x_j$ are $2,1,1,2$, and $2,1,1,2$
respectively. Then, 
it is easy to see that the price of the bundle of the clause interval in this case
is $mn^2 + 6(i + j - 2) + 4$, exceeding the budget of the clause interval. In the other three
cases, the price of the bundle is at most $mn^2 + 6(i + j - 2) + 3$, and the profit 
from the clause interval is at least $mn^2 + 6(i+j-2) + 2$ (In the case when both $x_i$ and 
$x_j$ are TRUE, the profit is $mn^2 + 6(i+j-2) + 2$, in the two other satisfying assignments
the profit is $mn^2 + 6(i+j-2) + 3$).
The proofs for the other types of clauses \xcy , \xyc , and \xcyc\ 
are similar.\qed
\qed

\medskip

\noindent{\bf Proof of Theorem \ref{t-strongNP}}.~ 
Suppose the instance of MAX-2-SAT has $k$ satisfied clauses. We set the prices for the edges
corresponding to the two basic gadgets corresponding to the variable $x_i$ to TRUE if $x_i = 1$
and FALSE otherwise. We set the price of edge $h$ to $n^2$. This gives a total profit of 
\begin{eqnarray*}
S &=& T\cdot 36n + T\cdot \sum_{i=1}^n (mn^2 + 6(2i-2) + 6) + T\cdot mn^2 + kmn^2 + O(kn)\\
\end{eqnarray*}
The first term of the sum comes
from the basic gadgets of each variable set to TRUE or FALSE, the second term comes from the
consistency gadgets, the third term comes from the $H$ intervals, and the last two terms,
from the satisfied clause gadgets.

To show the reverse direction, consider a price assignment that achieves a profit of at least
$S$.
We claim that in an optimal
price assignment, the gadgets corresponding to the variables are all consistent, and the edge $h$
has a price of $mn^2$. Note first that the maximum profit we can gain from all the clauses
is $O(m^2n^2)$. Now, if we have larger than, say $T=m^3n^2$ copies of each variable gadget, 
it follows from Lemma \ref{lem:consistency} that 
we only lose by making either the variable gadgets inconsistent, or if the $H$ intervals
and the consistency gadgets do not spend their entire budget.
Hence, in the optimal solution, 
the variables are consistent, and $h$ has a price of $mn^2$. This then leaves only the clause intervals. 
Note that our profit maximizing pricing will try to maximize the number of clause intervals satisfied,
since the clause intervals differ by at most $O(n)$ in their budgets, but their individual budgets 
themselves are at least $mn^2$.
By the obvious assignment of truth values to the variables from the variable price assignment, 
we get an assignment that satisfies $k$ clauses.
\qed

\newpage

\noindent
{\bf\Large Appendix B: The gadgets used in the NP-hard construction}

\begin{figure}[htpb]
\begin{center}
\input{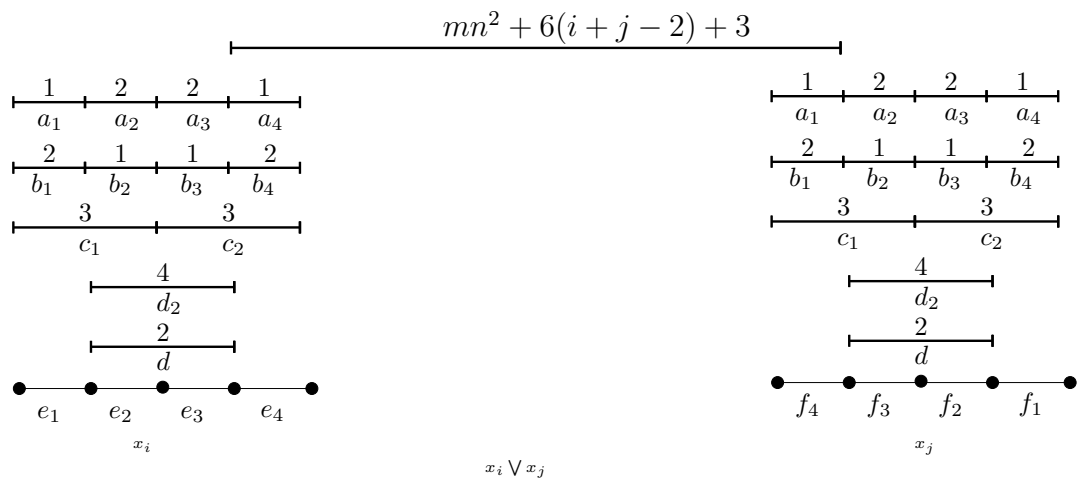}
\end{center}
\caption{The clause gadget for clause \xy}
\label{fig:prxy}
\end{figure}

\begin{figure}[htpb]
\begin{center}
\input{prxcy1}
\end{center}
\caption{The clause gadget for clause \xcy}
\label{fig:prxcy}
\end{figure}

\begin{figure}[htpb]
\begin{center}
\input{prxyc1}
\end{center}
\caption{The clause gadget for clause \xyc}
\label{fig:prxyc}
\end{figure}

\begin{figure}[htpb]
\begin{center}
\input{prxcyc1}
\end{center}
\caption{The clause gadget for clause \xcyc}
\label{fig:prxcyc}
\end{figure}

\end{document}